\documentclass[a4paper,runningheads]{llncs}

\usepackage{amssymb,amsmath,amsopn,amsthm}
\usepackage{mathtools}
\usepackage{graphicx}
\usepackage{authblk}
\usepackage{hyperref}
\usepackage[capitalise, noabbrev]{cleveref}
\usepackage{thm-restate}
\usepackage{multirow}
\usepackage{makecell}
\usepackage{subcaption}

\DeclareGraphicsExtensions{.pdf,.png,.eps}
\graphicspath{{figs/}}

\usepackage{csquotes}

\makeatletter
\g@addto@macro\bfseries{\boldmath}
\makeatother

\newcommand{\Self}{\textbf{S}} 
\newcommand{\Multi}{\textbf{M}} 
\newcommand{\Inc}{\textbf{I}} 
\newcommand{\Homoto}{\textbf{H}}

\usepackage{booktabs}

\newif\iflong
\newif\ifshort
\iflong
\else
\shorttrue
\fi

\usepackage{apptools}

\newcounter{section-preserve}
\newcounter{theorem-preserve}
\newcounter{lemma-preserve}
\newcounter{figure-preserve}
\newcommand{\blank}[1]{}
\newcommand{\short}[1]{\ifbool{long}{}{#1}}
\newcommand{\longonly}[1]{\ifbool{long}{#1}{}}
\newtoks\magicAppendix
\magicAppendix={}
\newtoks\magictoks
\newif\iflater
\laterfalse
\ifbool{long}{%
	\newcommand{\later}[1]{#1}%
}%
{%
\long\def\later#1{\magictoks={#1}%
	\edef\magictodo{\noexpand\magicAppendix={\the\magicAppendix%
			\the\magictoks%
	}}%
	\magictodo}%
}%
\long\def\latertitle#1{\magictoks={\section{#1}}%
	\edef\magictodo{\noexpand\magicAppendix={\the\magicAppendix \par
			\the\magictoks%
	}}%
	\magictodo}

\long\def\both#1{\magictoks={#1}%
	\edef\magictodo{\noexpand\magicAppendix={\the\magicAppendix%
			\noexpand\setcounter{theorem-preserve}{\noexpand\arabic{theorem}}%
			\noexpand\setcounter{theorem}{\arabic{theorem}}%
            \noexpand\setcounter{lemma-preserve}{\noexpand\arabic{lemma}}%
			\noexpand\setcounter{lemma}{\arabic{lemma}}%
			\noexpand\setcounter{section-preserve}{\noexpand\arabic{section}}%
			\noexpand\setcounter{section}{\arabic{section}}%
			\noexpand\let\noexpand\oldsection=\noexpand\thesection%
			\noexpand\def\noexpand\thesection{\thesection}%
			\the\magictoks%
			\noexpand\setcounter{theorem}{\noexpand\arabic{theorem-preserve}}%
            \noexpand\setcounter{lemma}{\noexpand\arabic{lemma-preserve}}%
			\noexpand\setcounter{section}{\noexpand\arabic{section-preserve}}%
			\noexpand\let\noexpand\thesection=\noexpand\oldsection%
	}}%
	\magictodo%
	\the\magictoks}%
\def\magicappendix{\latertrue \the\magicAppendix}

\makeatletter
\newcommand{\shortandlong}[1]{%
	\ifbool{long}{%
		#1%
	}{%
		\IfAppendix{}%
		{%
		    #1%
		}%
	}%
}%
\makeatother

\makeatletter
\newcommand{\appendixandlong}[1]{%
	\ifbool{long}{%
		#1%
	}{%
		\IfAppendix{%
			#1%
		}%
		{}%
	}%
}%
\makeatother

\makeatletter
\newcommand{\appendixandlongorshort}[2]{%
	\ifbool{long}{%
		#1%
	}{%
		\IfAppendix{%
			#1%
		}%
		{%
			#2%
		}%
	}%
}%
\makeatother

\makeatletter
\newcommand{\longandshortorappendix}[2]{%
	\ifbool{long}{%
		#1%
	}{%
		\IfAppendix{%
			#2%
		}%
		{%
			#1%
		}%
	}%
}%
\makeatother

\makeatletter
\patchcmd{\thmhead}{\thmnote{ {\the\thm@notefont(#3)}}}%
{%
	\ifbool{long}{%
		\ifstrequal{#3}{($\star$)}{%
		}{%
			\thmnote{ {\the\thm@notefont(#3)}}%
		}%
	}{%
		\IfAppendix{%
			\ifstrequal{#3}{($\star$)}{%
			}{%
				\thmnote{ {\the\thm@notefont(#3)}}%
			}%
		}%
		{%
			\thmnote{ {\the\thm@notefont(#3)}}%
		}%
	}%
}
{}{} 
\makeatother

\tracingpatches
\makeatletter
\patchcmd{\@spopargbegintheorem}{\item[\hskip\labelsep{#4#1\ #2}]{#4(#3)\@thmcounterend\ }#5}%
{%
	\ifbool{long}{%
		\ifstrequal{#3}{$\star$}{%
    \item[\hskip\labelsep{#4#1\ #2\@thmcounterend}]#5%
		}{%
    \item[\hskip\labelsep{#4#1\ #2}]{#4(#3)\@thmcounterend\ }#5%
		}%
	}{%
		\IfAppendix{%
			\ifstrequal{#3}{$\star$}{%
      \item[\hskip\labelsep{#4#1\ #2\@thmcounterend}]#5%
			}{%
      \item[\hskip\labelsep{#4#1\ #2}]{#4(#3)\@thmcounterend\ }#5%
			}%
		}%
		{%
    \item[\hskip\labelsep{#4#1\ #2}]{#4(#3)\@thmcounterend\ }#5%
		}%
	}%
}
{}{} 
\makeatother

\title{Edge-Minimum Saturated $k$-Planar Drawings}

\newbox{\myorcidauthbox}
\sbox{\myorcidauthbox}{\large\includegraphics[height=1.7ex]{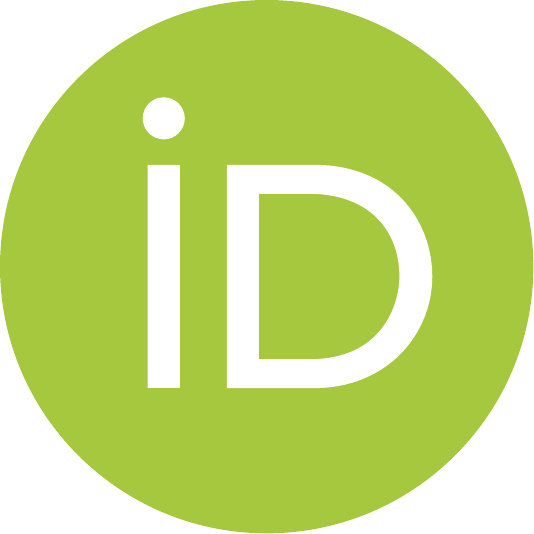}}
\newcommand{\orcid}[1]{\href{https://orcid.org/#1}{\usebox{\myorcidauthbox}}}

\author{Steven Chaplick\inst{1}\orcid{0000-0003-3501-4608} \and 
Fabian Klute\inst{2}\orcid{0000-0002-7791-3604} \and 
Irene Parada\inst{3}\orcid{0000-0003-3147-0083} \and 
Jonathan Rollin\inst{4}\orcid{0000-0002-6769-7098} \and 
Torsten Ueckerdt\inst{5}\orcid{0000-0002-0645-9715}}

\institute{Department of Data Science and Knowledge Engineering, Maastricht University, The Netherlands\\ \email{s.chaplick@maastrichtuniversity.nl} \and
Utrecht University, The Netherlands \\
	\email{f.m.klute@uu.nl}\and
  TU Eindhoven, The Netherlands\\
	\email{i.m.de.parada.munoz@tue.nl} \and
Department of Mathematics and Computer Science, FernUniversität in Hagen, Germany \\ \email{jonathan.rollin@fernuni-hagen.de} \and
  Institute of Theoretical Informatics, Karlsruhe Institute of Technology, Germany \\ \email{torsten.ueckerdt@kit.edu}}

\authorrunning{S. Chaplick, F. Klute, I. Parada, J. Rollin, T. Ueckerdt}

\begin{document}

\maketitle

\begin{abstract}
For a class $\mathcal{D}$ of drawings of loopless (multi-)graphs in the plane, a drawing $D \in \mathcal{D}$ is \emph{saturated} when the addition of any edge to $D$ results in $D' \notin \mathcal{D}$---this is analogous to saturated graphs in a graph class as introduced by Tur\'an (1941) and Erd\H{o}s, Hajnal, and Moon (1964). 
We focus on $k$-planar drawings, that is, graphs drawn in the plane where each edge is crossed at most $k$ times, and the classes $\mathcal{D}$ of all $k$-planar drawings obeying a number of restrictions, such as having no crossing incident edges, no pair of edges crossing more than once, or no edge crossing itself.

While saturated $k$-planar drawings are the focus of several prior works, tight bounds on how sparse these can be are not well understood.
We establish a generic framework to determine the minimum number of edges among all $n$-vertex saturated $k$-planar drawings in many natural classes.
For example, when incident crossings, multicrossings and selfcrossings are all allowed, the sparsest $n$-vertex saturated $k$-planar drawings have $\frac{2}{k - (k \bmod 2)} (n-1)$ edges for any $k \geq 4$, while if all that is forbidden, the sparsest such drawings have $\frac{2(k+1)}{k(k-1)}(n-1)$ edges for any $k \geq 6$.\par
\end{abstract}

\section{Introduction}

Graph \emph{saturation problems} concern the study of edge-extremal $n$-vertex graphs under various restrictions. 
They originate in the works of Tur\'an~\cite{Turan1941} and Erd\H{o}s, Hajnal, and Moon~\cite{EHM1964}. 
For a family $\mathcal{F}$ of graphs, a graph $G$ without loops or parallel edges is called \emph{$\mathcal{F}$-saturated} when no subgraph of $G$ belongs to $\mathcal{F}$ and for every $u,v\in V(G)$, where $uv \notin E(G)$, some subgraph of the graph $G+uv$ belongs to~$\mathcal{F}$. 
Tur\'an~\cite{Turan1941} described, for each $t$, the $n$-vertex graphs that are $\{K_t\}$-saturated and have the maximum number of edges---this led to the introduction of the Tur\'an Numbers where the setting moves from graphs to hypergraphs, see for example the surveys~\cite{keevash_2011,Sidorenko95}. 
Analogously, Erd\H{o}s, Hajnal, and Moon~\cite{EHM1964} studied the $n$-vertex graphs $G$ that are $\{K_t\}$-saturated and have the minimum number of edges. 
This sparsest saturation view has also received much subsequent study~\cite{Faudree2011}, 
and our work fits into this latter direction but concerns ``drawings of (multi-)graphs'', also called \emph{topological (multi-)graphs}. 

There has been increasing interest in saturation problems on drawings of (multi-)graphs in addition to the abstract graphs above. 
A \emph{drawing} is a graph together with a cyclic order of edges around each vertex and the sequence of crossings along each edge so that it can be realized in the plane (or on another specified surface). 
The saturation conditions usually concern the crossings (which can be thought of as avoiding certain topological subgraphs). 
The majority of work has been on \emph{Tur\'an-type} results regarding the maximum number of edges which can occur in an $n$-vertex drawing (without loops and homotopic parallel edges) of a particular drawing style, e.g., $n$-vertex planar drawings are well known to have at most $3n-6$ edges for any $n \geq 3$. 
In the case of planar drawings (i.e., crossing-free in the plane), the sparsest saturation version (as in Erd\H{o}s, Hajnal, and Moon~\cite{EHM1964}) is also equal to the Tur\'an version: Every saturated planar drawing has $3n-6$ edges.

\begin{figure}[t]
 \centering
 \includegraphics{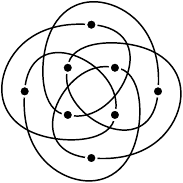}
 \hspace{2em}
 \includegraphics{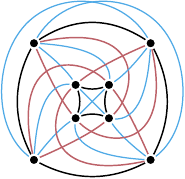}
 \hspace{2em}
 \includegraphics{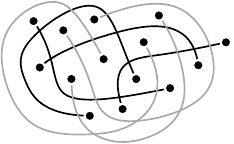}
	\caption{Saturated $4$-planar drawing of the $8$-cycle (left), $3$-planar drawing of the $8$-clique (middle), and  saturated $6$-planar drawing of the $7$-matching (right).}
 \label{fig:k-planar-examples}
\end{figure} 

However, for drawing styles that allow crossings in a limited way, these two measures become non-trivial to compare and can indeed be quite different.
This interesting phenomenon happens for example for \emph{$k$-planar} drawings where at most $k$ crossings on each edge are allowed; and which are the focus of the present paper.
The left of \cref{fig:k-planar-examples} depicts a drawing of the $8$-cycle $C_8$ in which each edge is crossed exactly four times and one cannot add a ninth (non-loop) edge to the drawing while maintaining $4$-planarity, i.e., this is a saturated $4$-planar drawing of $C_8$.
On the other hand, note that even the complete graph $K_8$ in fact admits $3$-planar drawings as shown in the middle of \cref{fig:k-planar-examples}.

In this sense, we call a drawing that attains the Tur\'an-type maximum number of edges a \emph{max-saturated}\footnote{Sometimes these drawings are called \emph{optimal} in the literature~\cite{BekosKR17}.} drawing,
while a sparsest saturated drawing is called \emph{min-saturated}.
The target of this paper is to determine the number of edges in min-saturated $k$-planar drawings of loopless (multi-)graphs%
, i.e., the \emph{smallest} number of edges among all saturated $k$-planar drawings with $n$ vertices.
The answer will always be of the form $\alpha_k \cdot (n-1)$.
However, it turns out that the precise value of $\alpha_k$ depends on numerous subtleties of what precisely we allow in the considered $k$-planar drawings.
Such subtleties are formalized by drawing styles $\Gamma$, each one with its own constant $\alpha_\Gamma$.
As we always require $k$-planarity, we omit $k$ from the notation  $\alpha_\Gamma$.

For example, restricting to connected graphs, we immediately have at least $n-1$ edges on $n$ vertices, i.e., $\alpha_\Gamma \geq 1$.
And in fact we also have $\alpha_\Gamma \leq 1$ for all $k \geq 4$ as testified by entangled drawings of cycles like in the left of \cref{fig:k-planar-examples}.
Allowing disconnected graphs but restricting to contiguous drawings, we immediately have $\alpha_\Gamma \geq 1/2$ since we have minimum degree at least $1$ in that case.
And again we also have $\alpha_\Gamma \leq 1/2$ for all $k \geq 6$ as one can find saturated $k$-planar drawings of matchings like in the right of \cref{fig:k-planar-examples}.
Other subtleties occur when we distinguish whether selfcrossing edges, repeatedly crossing edges, crossing incident edges, etc., are allowed or forbidden.
We enable a concise investigation of all possible combinations by first deriving lower bounds on $\alpha_\Gamma$ for any drawing style that satisfies only some mild assumptions.
We can then consider each drawing style $\Gamma$ and swiftly determine the exact value of $\alpha_\Gamma$, thus determining the smallest number of edges among all $k$-planar drawings of that style on $n$ vertices.
Our results for multigraphs are summarized in \cref{tab:results}.

\subsection{Related Work}
\label{sec:related-work}

For $k$-planar graphs the Tur\'an-type question, the edge count in max-saturated drawings, is well studied.
Any $k$-planar \emph{simple}\footnote{A drawing is simple if any two edges share at most one point. In particular there are no parallel edges.} drawing on $n$ vertices contains at most $3.81\sqrt{k}n$ edges~\cite{Ackerman19}, and better (and tight) bounds are known for small $k$~\cite{Ackerman19,PachRTT06,PachT97}.
Specifically $1$-planar drawings contain at most $4n-8$ edges which is tight~\cite{PachT97}.
For $k\leq 3$, any $k$-planar drawing with the fewest crossings (among all $k$-planar drawings of the abstract graph) is necessarily simple~\cite{PachRTT06}.
Therefore the tight bounds for $k\leq 3$ also hold for drawings that are not necessarily simple.
However, already for $k=4$, Schaefer~\cite[p. 58]{SchaeferDS2020} has constructed $k$-planar graphs having no $k$-planar simple drawings, and these easily generalize to all $k > 4$.
Pach et al.~\cite{PachRTT06} conjectured that for every $k$ there is a max-saturated $k$-planar graph with a simple $k$-planar drawing. 
For $k=2,3$, the max-saturated $k$-planar homotopy-free multigraphs have been characterized~\cite{BekosKR17}.

In the sparsest saturation setting not only min-saturated $k$-planar drawings are of interest but also min-saturated $k$-planar (abstract) graphs: sparse $k$-planar graphs that are no longer $k$-planar after adding any edge~\cite{AuerBGH12,BaratT18,BrandenburgEGGHR12,EadesHKLSS13}.
The questions we address in this work have also been explicitly asked~\cite[Section 3.2]{bp-dagstuhlreport-hong_et_al:DR:2017:7038}.

Recently, the case of saturation problems for simple drawings has come into focus.
The Tur\'an-type question is trivial here as all complete graphs have simple drawings.
However, 
there are constructions of saturated simple drawings (and generalizations thereof) with only $O(n)$ edges~\cite{HajnalIRS18,KynclPRT15}.

\begin{table}[t]
 \centering
 \begin{tabular}{|c|c|c|c|}
 
 \hline
 
 $k$ & restrictions & \makecell{minimum number of\\ edges of saturated $k$-planar\\ drawings on $n$ vertices} & \makecell{tight\\ example\footnotemark}\\
 
 \hline
 \hline
 
 \multirow{2}{*}{$k \geq 4$} & no restriction & \multirow{2}{*}{$\frac{2}{k-(k \bmod 2)}\cdot (n-1)$} & \multirow{2}{*}{\cref{fig:tight-I-k4}} \\
 \cline{2-2}
 & $\Inc$ no incident crossings & & \\
 
 \hline
 
 \multirow{2}{*}{$k \geq 4$} & $\Self$ no selfcrossings & \multirow{2}{*}{$\frac{2}{k-1}\cdot (n-1)$} & \multirow{2}{*}{\cref{fig:tight-SI-S-k4}} \\
 \cline{2-2}
 & $\Self$ no self- and $\Inc$ no incident crossings & & \\
 
 \hline
 
 $k \geq 4$ & $\Multi$ no multicrossings & $\frac{2(k-1)}{(k-1)(k-2)+2}\cdot (n-1)$ & \cref{fig:tight-M-k4} \\
 
 \hline
 
 $k \geq 4$ & $\Self$ no self- and $\Multi$ no multicrossings & $\frac{2(k+1)}{k(k-1)}\cdot (n-1)$ & \cref{fig:tight-MS-k4} \\
 
 \hline
 
 $k = 4$ & \multirow{2}{*}{$\Inc$ no incident and $\Multi$ no multicrossings} & $\frac{4}{5}\cdot (n-1)$ & \multirow{2}{*}{\shortstack{\cref{fig:tight-MI-k45},\\ \cref{fig:tight-MI-k6}}} \\
 \cline{1-1}\cline{3-3}
 $k \geq 5$ & & $\frac{2(k-1)}{(k-1)(k-2)+2}\cdot (n-1)$ & \\
 
 \hline
 
\multirow{4}{*}{$k \geq 6$} & $\Self$ no self-, $\Multi$ no multi-, & \multirow{4}{*}{$\frac{2(k+1)}{k(k-1)}\cdot (n-1)$} & \multirow{4}{*}{\shortstack{\cref{fig:k-planar-examples}\\
 (right),\\ \cref{fig:tight-SIM-k7}}} \\
 & and $\Inc$ no incident crossings &  & \\
 
 \cline{2-2}
  
 & $\Self$ no self-, $\Multi$ no multi-, $\Inc$ no incident &  &  \\
 & crossings, and $\Homoto$ homotopy-free & & \\

\hline
 \end{tabular}
 \caption{Overview of results (see also \cref{thm:results}): The minimum number of edges of saturated $k$-planar drawings on $n$ vertices of a drawing style defined by a set of restrictions.}
 \label{tab:results}
\end{table}

\footnotetext{To attain the stated bound via these constructions, insert an isolated vertex in each empty cell.}

\subsection{Drawings, Crossing Restrictions, and Drawing Types}
\label{sec:prelim}

Throughout the paper, we consider topological drawings in the plane, that is, vertices are represented by distinct points in $\mathbb{R}^2$ and edges are represented by continuous curves connecting their respective endpoints.
We allow parallel edges but forbid loops.
As usual, edges do not pass through vertices, any two edges have only finitely many interior points in common, each of which is a proper crossing, and no three edges cross in a common point.
An edge may cross itself but it uses any crossing point at most twice.
Also, each of these selfcrossings are counted twice when considering the number of times that edge is crossed.

The \emph{planarization} of a drawing $D$ is the planar drawing obtained from $D$ by making each crossing into a new vertex, thereby subdividing the edges involved in the crossing.
Although we forbid loops in $D$ its planarization might have loops due to selfcrossing edges.
In a drawing, an edge involved in at least one crossing is a \emph{crossed} edge, while those involved in no crossing are the \emph{planar} or \emph{uncrossed} edges.
The \emph{cells} of a drawing are the connected components of the plane after the removal of every vertex and edge in $D$.
In other words, the cells of $D$ are the faces of its planarization.
A vertex $v$ is incident to a cell $c$ if $v$ is contained in the closure of~$c$, i.e., one could at least start drawing an uncrossed edge from $v$ into cell $c$.

Two distinct parallel edges $e$ and $f$ in a drawing $D$ are called \emph{homotopic}, if there is a homotopy of the sphere between $e$ and~$f$, that is, the curves of $e$ and $f$ can be continuously deformed into each other along the surface of the sphere while all vertices of $D$ are treated as holes.

In what follows, we investigate drawings that satisfy a specific set of restrictions, where we focus on those with frequent appearance in the literature:
\begin{itemize}
    \item \emph{$k$-planar}: Each edge is crossed at most $k$ times.
	
	\item \Homoto\ \emph{homotopy-free}: No two distinct parallel edges are homotopic.

	\item \Multi\ \emph{single-crossing}: Any pair of edges crosses at most once and any edge crosses itself at most once (edges with $t\in\{0,1,2\}$ common endpoints have at most $t+1$ common points).
	
	\item \Inc\ \emph{locally starlike}\footnote{In other papers this is also called \emph{star simple} or \emph{semi simple}~\cite{BFK15_crossing,starsimple20} and may not allow selfcrossing edges.}: Incident edges do not cross (selfcrossing edges are allowed).
	
	\item \Self\ \emph{selfcrossing-free}: No edge crosses itself.

	\item \emph{branching}: The drawing is $\Multi$ single-crossing, $\Inc$ locally starlike, $\Self$ selfcrossing-free, and $\Homoto$~homotopy-free.
\end{itemize}

A \emph{drawing style} is just a class $\Gamma$ of drawings, i.e., a predicate whether any given drawing $D$ is in $\Gamma$ or not.
A drawing style $\Gamma$ is \emph{monotone} if removing any edge or vertex from any drawing $D \in \Gamma$ results again in a drawing $D' \in \Gamma$, i.e., $\Gamma$ is closed under edge/vertex removal.

We consider drawing styles given by all $k$-planar drawings of finite, loopless multigraphs obeying a subset $X$ of the restrictions above.  
Such a drawing style is denoted by~$\Gamma_X$.
We focus on the restrictions \Multi\ forbidding multicrossings, \Self\ forbidding selfcrossings, \Inc\ forbidding incident crossings, and \Homoto\ forbidding homotopic edges.
Note that the $k$-planar drawing style is monotone, and so is $\Gamma_X$ for each $X \subseteq \{\Self,\Inc,\Multi\}$.
However, the style of all homotopy-free drawings is not monotone, as removing a vertex may render two edges homotopic.

We are interested in $k$-planar drawings in $\Gamma_X$ to which no further edge can be added without either violating $k$-planarity or any of the restrictions in $X$, and particularly in how sparse these drawings can be; namely, the sparsest saturated such drawings.

\begin{definition}
 A drawing $D$ is \emph{$\Gamma$-saturated} for drawing style $\Gamma$ if $D \in \Gamma$ and the addition of any new edge to $D$ results in a drawing $D' \notin \Gamma$.
\end{definition}

\subsection{Our Results}
In order to determine the sparsest $k$-planar $\Gamma_X$-saturated drawings for restrictions in $X$, we introduce in \cref{sec:lower-bounds-general} the concept of \emph{filled} drawings in general monotone drawing styles and give lower bounds on the number of edges in these.
Using the lower bounds for filled drawings and constructing particularly sparse $\Gamma_X$-saturated drawings, we then give in \cref{sec:exact-bounds} the precise answer for all $X \subseteq \{\Self,\Inc,\Multi\}$ and for the branching style, i.e., $X = \{\Self,\Inc,\Multi,\Homoto\}$, leaving open only a few cases for $k \in \{4,5,6\}$.
Our results for multigraphs are summarized in \cref{tab:results} and formalized in \cref{thm:results}.
In \cref{sec:simplegraphs} we discuss saturated drawings of simple graphs instead of multigraphs.
Finally, in \cref{sec:conclusions} we discuss further extensions.

{\bigskip\noindent\itshape Proofs of statements marked with {\normalfont($\star$)} can be found in the appendix.}

\section{Lower Bounds and Filled Drawings}
\label{sec:lower-bounds-general}
\latertitle{Omitted Proofs of \cref{sec:lower-bounds-general}}

Throughout this section, let $\Gamma$ be an arbitrary monotone drawing style; not necessarily $k$-planar or defined by any of the restrictions in \cref{sec:prelim}.
Recall that $\Gamma$ is monotone if 
it is closed under the removal of vertices and/or edges.

\begin{definition}\label{def:filled}
 A drawing $D$ is \emph{filled} if any two distinct vertices that are incident to the same cell $c$ of $D$ are connected by an uncrossed edge that lies completely in the boundary of $c$.
\end{definition}

For example, the filled crossing-free homotopy-free drawings are exactly the planar drawings of loopless multigraphs with every face bounded by three edges.
Using Euler's formula, such drawings on $n \geq 3$ vertices have exactly $m = 3n-6$ edges.
In this section we derive lower bounds on the number of edges in $n$-vertex filled drawings in drawing style $\Gamma$.
Another important example of filled drawings are those in which every cell has at most one incident vertex.
Note that every cell in a filled drawing has at most three incident vertices.
Generally, for a drawing $D$ we use the following notation:

\noindent
\begin{minipage}{0.23\textwidth}
 \begin{align*}
  n_D &= \#\text{ vertices}\\
  m_D &= \#\text{ edges} 
 \end{align*}
\end{minipage}
\begin{minipage}{0.77\textwidth}
 \begin{align*}
  c_i(D) &= \#\text{ cells with exactly $i$ incident vertices, $i \geq 0$} \\ 
  c'_2(D) &= \#\text{ cells with 2 uncrossed edges in their boundary} 
 \end{align*}
\end{minipage}
\smallskip

For a drawing $D$, let $G$ be its graph and $P$ be its planarization.
A \emph{component} of $D$ is a connected component of $P$.
A \emph{cut-vertex} of $D$ is a cut-vertex of $G$ that is also a cut-vertex of $P$.
And finally, $D$ is \emph{essentially $2$-connected} if 
 one component has at least one edge, all other components are isolated vertices and along the boundary of each cell each vertex appears at most once (that is, $D$ has no cut-vertex).
This means that for each simple closed curve that intersects $D$ in exactly one vertex or not at all, either the interior or the exterior contains no edges from $D$.

\both{\begin{lemma}[$\star$]\label{lem:wlog-2-connected}
 For every monotone drawing style $\Gamma$ and every filled drawing $D \in \Gamma$ we have 
$  m_D \geq \alpha_\Gamma \cdot (n_D+c_0(D)-1)$  where
\[
  \alpha_\Gamma = \min\left\{ \frac{m_{D'}}{n_{D'}+c_0(D')-1} \colon D' \in \Gamma \text{ is filled and essentially $2$-connected }\right\}.
\]
\end{lemma}}
\later{\begin{proof}
 We proceed by induction on the number $n_D$ of vertices in $D$.
 The desired inequality $m_D \geq \alpha_\Gamma(n_D + c_0(D)-1)$ clearly holds if $D$ itself is essentially $2$-connected.
 Otherwise $D$ has a cut vertex or two components with an edge.
 In both cases we can choose a simple closed curve $C$ with at least one vertex of $D$ in its interior and at least one vertex in its exterior such that $D \cap C$ is either empty or a single vertex.
 Let $D'$ and $D''$ denote the drawings obtained from $D$ by removing every edge and vertex of $D$ in the exterior of $C$, respectively interior of $C$.
 Observe that $D',D''$ are filled and in $\Gamma$, as $\Gamma$ is monotone.
 Further observe that $m_{D'} + m_{D''} = m_D$, as every edge of $D$ lies on one side of $C$.
 
 Now if $C \cap D \neq \emptyset$, then $C \cap D$ consists of exactly one vertex and $n_{D'} + n_{D''} = n_D + 1$.
 Moreover $c_0(D') + c_0(D'') = c_0(D)$, since the vertex in $C \cap D$ is incident to the cell containing curve $C$ in both drawings $D'$ and $D''$.
 Hence, using induction on $D'$ and $D''$ we conclude
 \begin{align*}
  m_D &= m_{D'} + m_{D''} \geq \alpha_\Gamma(n_{D'}+c_0(D')-1) + \alpha_\Gamma(n_{D''} + c_0(D'') -1)\\
  &= \alpha_\Gamma(n_{D'}+n_{D''} - 1 + c_0(D') + c_0(D'') - 1) = \alpha_\Gamma(n_D + c_0(D) - 1).
 \end{align*}

 On the other hand, if $C \cap D = \emptyset$, then $n_{D'} + n_{D''} = n_D$.
 Moreover $c_0(D') + c_0(D'') = c_0(D) + 1$, since the cell of $D$ containing curve $C$ can have incident vertices only on one side of $C$, as the drawing is filled.
 Similar as before, we conclude
 \begin{align*}
  m_D &= m_{D'} + m_{D''} \geq \alpha_\Gamma(n_{D'}+c_0(D')-1) + \alpha_\Gamma(n_{D''} + c_0(D'') -1)\\
  &= \alpha_\Gamma(n_{D'}+n_{D''} + c_0(D') + c_0(D'') - 1 - 1) \geq \alpha_\Gamma(n_D + c_0(D) - 1).\qedhere
 \end{align*}
\end{proof}}

As suggested by \cref{lem:wlog-2-connected}, we  shall now focus on filled drawings that are essentially $2$-connected.
Our goal is to determine the parameter $\alpha_\Gamma$.
First, we give an exact formula for the number of edges in any filled essentially $2$-connected drawing.
The parameter $k$ in the following lemma will later be the $k$ for the $k$-planar drawings in \cref{sec:exact-bounds}.
However, we do not require any drawing to be $k$-planar here.

\both{\begin{lemma}[$\star$]\label{lem:2-connected-count}
 For any $k > 2$, if $D$ is a filled, essentially $2$-connected drawing with $n_D \geq 3$ vertices, then
  $ m_D = \tfrac{2}{k-2}(n_D + c_0(D) - 2 + \varepsilon(D)),  \mbox{ where}$
 \begin{align*}
  \varepsilon(D) &= (\tfrac{k}{2}m_{\rm x} - \mathrm{cr}) + \tfrac{k-4}{4}m_{\rm p} + c'_2 + c_3, \qquad \mbox{such that}\\
  m_{\rm p} &= \#\text{planar edges}, \qquad  {\rm cr} = \#\text{crossings}, \mbox{and} \qquad  m_{\rm x} = \#\text{crossed edges}.
 \end{align*}%
\end{lemma}}
\later{\begin{proof}
 First observe that, since $D$ is filled, no cell has four or more incident vertices.
 Hence, $\#{\rm cells} =c_0+c_1+c_2+c_3$.
 By counting along the angles around each vertex, we see that
 \begin{equation}
  \#{\rm isolated} + 2m_D = \#{\rm isolated} + \sum_{v} \deg(v) = c_1 + 2c_2 + 3c_3.\label{eq-tight:cells-by-angle}
 \end{equation}
 Note that this relies on the assumption that $D$ is essentially $2$-connected, as this guarantees that each non-isolated vertex $v$ lies on the boundary of exactly $\deg(v)$ cells. 

 As $D$ is filled, each cell with exactly two vertices on its boundary is incident to either one or two planar edges and each cell with three vertices on its boundary is incident to exactly three planar edges.
 Moreover, each planar edge is contained in the boundary of exactly two distinct such cells since $D$ has no cut-vertices and $n_D \geq 3$.
 By counting along the sides of the planar edges, we see that
  $2m_{\rm p} = c_2 + c'_2 + 3c_3,\notag$
 which together with~\eqref{eq-tight:cells-by-angle} gives
 \begin{equation}
  \#{\rm isolated} + 2m_{\rm x} = c_1+c_2-c'_2. 
  \label{eq-tight:cell-count}
 \end{equation}

 Consider the planarization $P$ of $D$.
 Since $D$ is essentially $2$-connected, $P$ has exactly $(1+\#{\rm isolated})$ many connected components.
 Moreover we have
 \begin{equation}
  |V(P)|= n_D+\mathrm{cr} \quad\text{and}\quad |E(P)| = m_D + 2\mathrm{cr} \quad\text{and}\quad \#{\rm cells}=c_0 + c_1 + c_2 + c_3.\label{eq-tight:segment-count}
 \end{equation}
 Applying Euler's formula to $P$ we have
 \begin{align*}
  2 &= |V(P)| - |E(P)| + \#{\rm cells} - \#{\rm isolated} \\
  &\overset{\eqref{eq-tight:segment-count}}{=} \mathrm{cr} + n_D - m_D - 2\mathrm{cr} + c_0 + c_1 + c_2 + c_3 - \#{\rm isolated}\\
  &\overset{\eqref{eq-tight:cell-count}}{=} n_D - m_D - \mathrm{cr} + 2m_{\rm x} + \#{\rm isolated} + c'_2 + c_0 + c_3 - \#{\rm isolated} \\
  &= n_D + m_{\rm x} - m_{\rm p} - \mathrm{cr} + c'_2 + c_0 + c_3 \\
  &= n_D+\tfrac{2-k}{2}(m_{\rm x} + m_{\rm p}) + \tfrac{k-4}{2}m_{\rm p} + (\tfrac{k}{2}m_{\rm x}-\mathrm{cr}) + c'_2 + c_0 + c_3.%
 \end{align*}

 Solving for $m_D$ we have: 
 \[
  m_D = \frac{2}{k-2} \cdot \left(n_D+c_0-2 + (\tfrac{k}{2}m_{\rm x}-\mathrm{cr}) + \tfrac{k-4}{4} m_{\rm p} + c'_2 + c_3\right).\qedhere
 \]
\end{proof}}

\cref{lem:wlog-2-connected,lem:2-connected-count} together imply that for any filled drawing $D \in \Gamma$ we have
\begin{align*}
 \frac{m_D}{n_D-1} \geq \frac{m_D}{n_D+c_0(D)-1} &\geq \min_{D'} \frac{m_{D'}}{n_{D'}+c_0(D')-1}\\
 &= \min_{D'} \frac{2}{k-2} \cdot \frac{n_{D'}+c_0(D')-2 + \varepsilon(D')}{n_{D'}+c_0(D')-1 \hspace{3.5em}},
\end{align*}
where both minima are taken over all filled, essentially $2$-connected drawings $D' \in \Gamma$ and $\varepsilon(D')$ can be thought of as an error term for the drawing $D'$, which we seek to minimize.
Indeed, if $D'$ is $k$-planar, i.e., each edge is crossed at most $k$ times, then $2\mathrm{cr} \leq km_{\rm x}$.
Thus for $k \geq 4$ we have $\varepsilon(D') \geq 0$.
In the next section we shall see that (in many cases) the minimum is indeed attained by drawings $D'$ with $\varepsilon(D') = 0$.

\section{Exact Bounds and Saturated Drawings}
\label{sec:exact-bounds}
\latertitle{Omitted Proofs of \cref{sec:exact-bounds}\label{apx:results}}

Recall that we seek to find the sparsest $k$-planar, $\Gamma_X$-saturated drawings in a drawing style $\Gamma_X$ that is given by a set $X \subseteq \{\Self,\Inc,\Multi,\Homoto\}$ of additional restrictions.
These $\Gamma_X$-saturated drawings are related to the filled drawings from \cref{sec:lower-bounds-general}.

\both{\begin{lemma}[$\star$]\label{lem:saturated-are-filled}
 For any $k \geq 0$ and any $X \subseteq \{\Self,\Inc,\Multi\}$, as well as for $X = \{\Self,\Inc,\Multi,\Homoto\}$, every $k$-planar, $\Gamma_X$-saturated drawing is filled.
\end{lemma}}
\later{\begin{proof}
 Consider a $k$-planar, drawing $D \in \Gamma_X$ and a cell $c$ in $D$ with two incident vertices $u,v$, such that $u,v$ are not connected by an uncrossed edge in the boundary of $c$.
 That is, $D$ is not filled and we shall show that it is not $\Gamma_X$-saturated.
 We add a new uncrossed edge $e = uv$ in that cell, resulting in a new drawing $D'$.
 Clearly, the introduction of $e$ did not create any new selfcrossings, incident crossings, multicrossings, or edges being crossed more than $k$ times.
 Hence, for $X \subseteq \{\Self,\Inc,\Multi\}$, drawing $D'$ lies in $\Gamma_X$ and $D$ was not $\Gamma_X$-saturated.
 
 It remains to consider $X = \{\Self,\Inc,\Multi,\Homoto\}$ and rule out that $e$ is homotopic to another edge in order to show that $D'\in\Gamma_X$.
 So let $e'$ be an edge parallel to $e$ which is closest to $e$ in the cyclic order of edges incident to $u$.
 Since incident crossings and selfcrossings are forbidden, $e$ and $e'$ together form a simple closed curve $C$.
 If $e'$ is uncrossed, then $e'$ is not in the boundary of cell $c$.
 Since incident crossings are forbidden we find edges $f$ and $f'$ connecting $u$ to a vertex in the interior and a vertex in the exterior of $C$, respectively.
 Hence $e$ and $e'$ are not homotopic.
 On the other hand, suppose that $e'$ is crossed by some edge $e''$.
 As incident crossings are forbidden, neither $u$ nor $v$ is an endpoint of $e''$.
 As multicrossings are forbidden, the two endpoints of $e''$ lie in the exterior and the interior of $C$, respectively.
 Hence $e$ and $e'$ are not homotopic.
\end{proof}}

In order to determine the exact edge-counts for min-saturated drawings, we shall find for each drawing style some essentially $2$-connected, $\Gamma_X$-saturated drawings that attain the minimum in \cref{lem:wlog-2-connected}.
Motivated by the error term $\varepsilon(D) = (\frac{k}{2}m_{\rm x} - \mathrm{cr}) + \frac{k-4}{4}m_{\rm p} + c'_2 + c_3$ in \cref{lem:2-connected-count}, we define \emph{tight drawings} as those $k$-planar drawings in which \textbf{1)} every edge is crossed exactly $k$ times (so $\frac{k}{2}m_{\rm x} = \mathrm{cr}$) and \textbf{2)} every cell contains exactly one vertex (so $m_{\rm p}=c_0=c'_2=c_3=0$).
Observe that tight drawings are indeed $\Gamma_X$-saturated and filled and exist only in case $k\geq 4$.
Note that, to aid readability, isolated vertices are omitted from the drawings in the figures.
Namely, the actual drawings have one isolated vertex in each cell shown empty in the figures.
This is also mentioned in the figure captions.

\both{\begin{lemma}[$\star$]\label{lem:bound-from-tight}
 For every $k \geq 4$ and every monotone drawing style $\Gamma$ of $k$-planar drawings, if $D\in\Gamma$ is a tight drawing, then $\alpha_\Gamma \leq \frac{2}{k-2} \cdot \frac{n_D-2}{n_D-1} < 1$.
\end{lemma}}
\later{\begin{proof}
 If $D$ is not essentially $2$-connected, then there is a closed curve $C$ containing edges of $D$ in the interior as well as exterior, such that $C \cap D$ is either empty or a single vertex.
 Then the drawing obtained by removing everything inside $C$ (and adding an isolated vertex if the resulting cell is empty) is again in $\Gamma$ by monotonicity and again tight, but has fewer vertices.
 As $\frac{n-2}{n-1}$ is monotone increasing in $n$, it thus suffices to prove the claim for any essentially $2$-connected tight drawing $D_0$. 
 
 Clearly, $D_0$ is filled, as there are no two vertices incident to the same cell.
 We immediately get $c_0(D_0) = 0$, $m_{\rm p} = 0$, $c'_2 = c_3 = 0$, $2\mathrm{cr} = km_{\rm x}$, and it follows that $\varepsilon(D_0) = (\tfrac{k}{2}m_{\rm x} - \mathrm{cr}) + \tfrac{k-4}{4}m_{\rm p} + c'_2 + c_3 = 0$.
 As there is at least one edge in $D_0$ and this is crossed $k \geq 4$ times, Euler's formula implies that there are at least three cells.
 Hence $n_{D_0} \geq 3$ and \cref{lem:2-connected-count} gives
 \[
  \alpha_\Gamma \leq \frac{m_{D_0}}{n_{D_0}+c_0(D_0) - 1} = \frac{m_{D_0}}{n_{D_0}-1} = \frac{2}{k-2}\cdot \frac{n_{D_0} - 2}{n_{D_0} - 1} < \frac{2}{k-2} \leq 1.\qedhere
 \]
\end{proof}}

\both{\begin{theorem}[see also \cref{tab:results}]\label{thm:results}
 Let $k \geq 4$, $X \subseteq \{\Self,\Inc,\Multi,\Homoto\}$ be a set of restrictions, and $\Gamma = \Gamma_X$ be the corresponding drawing style of $k$-planar drawings. 
 
 For infinitely many values of $n$, the minimum number of edges in any $n$-vertex $\Gamma$-saturated drawing is 

\begin{center} \renewcommand{\arraystretch}{1.5}
\begin{tabular}{rl}
 $\frac{2}{k - (k \bmod 2)}(n-1)$ & \quad for $X = \{\Inc\}$ and $X = \emptyset$. \\
 $\frac{2}{k-1}(n-1)$ & \quad for $X = \{\Self\}$ and $X = \{\Self,\Inc\}$. \\  
  $\frac{2(k-1)}{(k-1)(k-2)+2}(n-1)$ & \quad for $X = \{\Multi\}$. \\
  $\frac{2(k+1)}{k(k-1)}(n-1)$ & \quad for $X = \{\Self,\Multi\}$. \\
  $\frac{4}{5}(n-1)$ & \quad for $X = \{\Inc,\Multi\}$ and $k=4$. \\  
  $\frac{2(k-1)}{(k-1)(k-2)+2}(n-1)$ & \quad for $X = \{\Inc,\Multi\}$ and $k\geq 5$. \\  
$\frac{2(k+1)}{k(k-1)}(n-1)$ & \quad for $X = \{\Self,\Inc,\Multi\}$ and $k \geq 6$. \\
$\frac{2(k+1)}{k(k-1)}(n-1)$ & \quad for $X = \{\Self,\Inc,\Multi,\Homoto\}$ and $k \geq 6$.
\end{tabular}
\end{center}
\end{theorem}}
\both{\begin{proof}
 \longandshortorappendix{For space requirements we present four out of seven cases in detail,
 the other cases can be found in Appendix~\ref{apx:results}.
 We start with the cases when $X \subseteq \{\Self,\Inc,\Multi\}$.
 Here the drawing style $\Gamma_X$ is monotone and every $\Gamma_X$-saturated drawing is filled by \cref{lem:saturated-are-filled}.
 Thus, by \cref{lem:bound-from-tight}, we have $\alpha_\Gamma \leq \frac{2}{k-2} \cdot \frac{n_{D_0}-2}{n_{D_0}-1}$ for every tight drawing $D_0$.
 This gives the smallest bound when $n_{D_0}$ is minimized. %
 In this case $D_0$ is essentially $2$-connected and $m_{D_0} = \frac{2}{k-2}(n_{D_0}-2)$ by \cref{lem:2-connected-count}, since $n_{D_0}\geq 3$ for tight drawings.
 So it suffices to consider a tight drawing $D_0$ with the smallest possible number $m_{D_0}$ of edges.
 
 Next, we shall go through the possible subsets $X$ of $\{\Self,\Inc,\Multi\}$ and determine exactly the value $\alpha_\Gamma$ for $\Gamma = \Gamma_X$ in two steps.
 \begin{itemize}
  \item First, we present a tight (hence filled) drawing $D_0$ with the smallest possible number $m_{D_0}$ of edges, which gives by \cref{lem:bound-from-tight} the upper bound 
  \[
   \alpha_\Gamma \leq \frac{2}{k-2}\cdot \frac{n_{D_0}-2}{n_{D_0}-1}.
  \]
  
  \item Second, we argue that for every filled (hence also every $\Gamma_X$-saturated), essentially $2$-connected drawing $D' \in \Gamma_X$ we have
  \begin{equation}
   \frac{n_{D'}+c_0(D')-2 + \varepsilon(D')}{n_{D'}+c_0(D')-1\hspace{3.5em}} \geq \frac{n_{D_0}-2}{n_{D_0}-1},
   \label{eq:LB-on-alpha}
  \end{equation}
  which by \cref{lem:wlog-2-connected,lem:2-connected-count} then proves the matching lower bound:
  \begin{align*}
   \alpha_\Gamma = \min_{D'} \frac{m_{D'}}{n_{D'}+c_0(D')-1} &= \min_{D'} \frac{2}{k-2} \cdot \frac{n_{D'}+c_0(D')-2 + \varepsilon(D')}{n_{D'}+c_0(D')-1\hspace{3.5em}} \\
   &\overset{\eqref{eq:LB-on-alpha}}{\geq} \frac{2}{k-2} \cdot \frac{n_{D_0}-2}{n_{D_0}-1}
  \end{align*}
 \end{itemize}

 In order to verify \eqref{eq:LB-on-alpha}, observe that if $\varepsilon(D') \geq 1$, then the lefthand side is at least $1$, while the righthand side is less than $1$.
 Thus it is enough to verify \eqref{eq:LB-on-alpha} when $\varepsilon(D') < 1$.
 In particular we may assume $c'_2 = c_3 = 0$ and $2\mathrm{cr} \geq km_{\rm x} - 1$ for $D'$.
 Similarly, as $\varepsilon(D') \geq 0$, we may assume that $n_{D'}+c_0(D') \leq n_{D_0}-1$.
 Altogether this implies that \eqref{eq:LB-on-alpha} is fulfilled unless
 \begin{equation*}
  m_{D'} = \frac{2}{k-2} (n_{D'}+c_0(D')-2+\varepsilon(D')) < \frac{2}{k-2}(n_{D_0}-1-2+1) = m_{D_0}.
 \end{equation*}
 In summary, for each $X$ we shall give a tight drawing $D_0$ with as few edges as possible, and argue that every filled, essentially $2$-connected drawing $D'$ with fewer edges satisfies the inequality \eqref{eq:LB-on-alpha}.
 Note that $m_{D'} \geq 1$ as essentially $2$-connected drawings have at least one edge.
 In fact, we may assume that $D'$ contains at least one crossed edge. 
 Otherwise $D'$ is filled, planar and hence connected.
 Thus $m_{D'} \geq n_{D'}-1$ and $c_0(D') = 0$ which verifies \eqref{eq:LB-on-alpha} as follows:
 \begin{align*}
  \frac{n_{D'}-c_0(D')-2+\varepsilon(D')}{ n_{D'}-c_0(D')-1\hspace{3.5em}} 
  = \frac{k-2}{2}\cdot  \frac{m_{D'}}{n_{D'}-1}\geq 1 > \frac{n_{D_0}-2}{n_{D_0}-1}
 \end{align*}}{%
 Here, we present the missing cases from the proof of \cref{thm:results}.
 }

 \begin{description}
  \shortandlong{%
  \item[Case 1. $X = \{\Inc\}$ and $X = \emptyset$]{\ \\}
  \cref{fig:tight-I-k4} shows drawings $D_0$ with $m_{D_0}=1$ edge when $k$ is even, and $m_{D_0} = 2$ edges when $k$ is odd, which are tight for $\Gamma = \Gamma_X$ for both $X = \{\Inc\}$ and $X = \emptyset$, as incident edges do not cross.
  Thus $m_{D_0} = 1 + (k \bmod 2)$ and $n_{D_0} = \frac{k+2}{2}$ for $k$ even, respectively $n_{D_0} = k$ for $k$ odd.
  Together this gives $\alpha_\Gamma \leq \frac{2}{k-2}\cdot \frac{n_{D_0}-2}{n_{D_0}-1} = \frac{2}{k - (k \bmod 2)}$.
  
  \begin{figure}[t]
   \centering
   \includegraphics{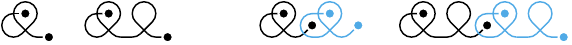}
   \caption{Smallest tight drawings for even $k\geq4$ (left) and odd $k\geq 4$ (right) in case $X = \emptyset$ and $X = \{\Inc\}$, i.e. nothing, resp. incident crossings, are forbidden. (Isolated vertices in empty cells are omitted.)}
   \label{fig:tight-I-k4}
  \end{figure}
 
  On the other hand, let $D' \in \Gamma_X$ be any filled, essentially $2$-connected drawing.
  As argued above, we may assume that $1 \leq m_{\rm x} \leq m_{D'} < m_{D_0}$.
  For even $k$, there is nothing to show as $m_{D'} \geq 1 = m_{D_0}$.
  For odd $k$, we may assume that $D'$ consists of exactly one edge, which has exactly $(k-1)/2$ selfcrossings (since $2\mathrm{cr} \geq km_{\rm x} - 1$), and some of the resulting cells may contain an isolated vertex.
  In particular, $\varepsilon(D') \geq \frac{k}{2}m_{\rm x} - \mathrm{cr} = 1/2$. 
  Applying Euler's formula to the planarization of $D'$ we get $n_{D'} + c_0(D') = (k+1)/2$, which verifies \eqref{eq:LB-on-alpha} as follows:
  \begin{align*}
    \frac{n_{D'}+c_0(D')-2 + \varepsilon(D')}{n_{D'}+c_0(D')-1\hspace{3.5em}} 
    \geq \frac{(k+1)/2 - 2 +1/2}{(k+1)/2 - 1}
    = \frac{k-2}{k-1} 
    = \frac{n_{D_0}-2}{n_{D_0}-1}.%
  \end{align*}}%
\appendixandlongorshort{%
  \item[Case 2. $X = \{\Self\}$ and $X = \{\Self,\Inc\}$]{\ \\}
  \cref{fig:tight-SI-S-k4} shows drawings $D_0$ with $m_{D_0}=2$ edges which are tight for $\Gamma = \Gamma_X$ for both $X = \{\Self\}$ and $X = \{\Self,\Inc\}$, as there are neither incident crossings nor selfcrossings.
  Thus $m_{D_0} = 2$ and $n_{D_0} = k$, which gives $\alpha_\Gamma \leq \frac{2}{k-2}\cdot \frac{n_{D_0}-2}{n_{D_0}-1} = \frac{2}{k - 1}$.
  
  \begin{figure}[t]
  \centering
  \includegraphics{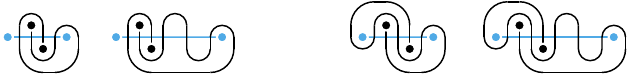}
  \caption{Smallest tight drawings for $k\geq4$ in case $X = \{\Self\}$ and $X = \{\Self,\Inc\}$, i.e. selfcrossings (resp. also incident crossings) are forbidden. (Isolated vertices in empty cells are omitted.)}
  \label{fig:tight-SI-S-k4}
  \end{figure}
  
  On the other hand, let $D'$ be any drawing in $\Gamma_X$, and assume again that $1 \leq m_{\rm x} < m_{D_0} = 2$.
  In particular, $D'$ has exactly one crossed edge, which however is impossible as selfcrossings are forbidden.%
  \smallskip
  
  \item[Case 3. $X = \{\Multi\}$]{\ \\}
  \cref{fig:tight-M-k4} shows tight drawings $D_0$ with $m_{D_0}=k-1$ edges.
  Thus $n_{D_0} = \frac{k-2}{2}m_{D_0} + 2 = \binom{k-1}{2}+2$, which gives
  \[
  \alpha_\Gamma \leq \frac{2}{k-2}\cdot \frac{n_{D_0}-2}{n_{D_0}-1} = \frac{2}{k-2} \cdot \frac{\binom{k-1}{2}}{\binom{k-1}{2}+1} = \frac{2(k-1)}{(k-1)(k-2)+2}.
  \]
  
  \begin{figure}[t]
  \centering
  \includegraphics{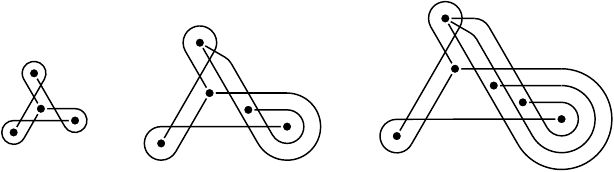}
  \caption{Smallest tight drawings for $k\geq4$ in case $X = \{\Multi\}$, i.e. multicrossings are forbidden. (Isolated vertices in empty cells are omitted.)}
  \label{fig:tight-M-k4}
  \end{figure}
  
  On the other hand, let $D'$ be any drawing in $\Gamma_X$.
  As argued above the desired inequality \eqref{eq:LB-on-alpha} holds, unless $km_{\rm x} -1 \leq 2\mathrm{cr}$ and $1 \leq m_{\rm x} \leq m_{D'} < m_{D_0} = k-1$.
  As there are no multicrossings, the crossed edges may pairwise cross at most once, and additionally each crossed edge may cross itself at most once, i.e., $\mathrm{cr} \leq m_{\rm x} + \binom{m_{\rm x}}{2} = \binom{m_{\rm x}+1}{2}$.
  However, this would imply
  \[
  km_{\rm x}-1 \leq 2\mathrm{cr} \leq (m_{\rm x}+1)m_{\rm x} \leq (k-2+1)m_{\rm x} = km_{\rm x} - m_{\rm x},
  \]
  and thus $m_{\rm x} = 1$.
  However, then $2\mathrm{cr} \geq km_{\rm x} - 1 = k-1 \geq 3$, which contradicts that there are no multicrossings.%
  \smallskip}{%
  \item[Case 2. $X = \{\Self\}$ and $X = \{\Self,\Inc\}$ ($\star$)]{\ }\smallskip
  \item[Case 3. $X = \{\Multi\}$ ($\star$)]{\ }\smallskip
  }
  
  \shortandlong{%
  \item[Case 4. $X = \{\Self,\Multi\}$]{\ \\}
  \cref{fig:tight-MS-k4} shows tight drawings $D_0$ with $m_{D_0}=k+1$ edges.
  Thus $n_{D_0} = \frac{k-2}{2}m_{D_0} + 2 = \binom{k}{2}+1$, which gives
  \[
   \alpha_\Gamma \leq \frac{2}{k-2}\cdot \frac{n_{D_0}-2}{n_{D_0}-1} = \frac{2}{k-2} \cdot \frac{\binom{k}{2}-1}{\binom{k}{2}} = \frac{2(k+1)}{k(k-1)}.
  \]
  
  \begin{figure}[t]
   \centering
   \includegraphics{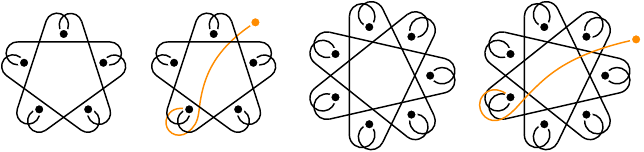}
   \caption{Smallest tight drawings for $k\geq4$ in case $X = \{\Self,\Multi\}$, i.e. selfcrossings and multicrossings are forbidden. (Isolated vertices in empty cells are omitted.)}
   \label{fig:tight-MS-k4}
  \end{figure}

  On the other hand, let $D'$ be any drawing in $\Gamma_X$.
  Again \eqref{eq:LB-on-alpha} holds, unless $km_{\rm x} -1 \leq 2\mathrm{cr}$ and $1 \leq m_{\rm x} \leq m_{D'} < m_{D_0} = k+1$.
  As there are no multicrossings and no selfcrossings, we have $\mathrm{cr} \leq \binom{m_{\rm x}}{2}$.
  However, this would imply $km_{\rm x}-1 \leq 2\mathrm{cr} \leq m_{\rm x}(m_{\rm x}-1) \leq k(m_{\rm x}-1) = km_{\rm x} - k \leq km_{\rm x}-4$, which is a contradiction.%
  \smallskip}
 
 \appendixandlongorshort{%
  \item[Case 5. $X = \{\Inc,\Multi\}$]{\ \\}
  \cref{fig:tight-MI-k45,fig:tight-MI-k6} show tight drawings $D_0$ with $m_{D_0}=4$ edges for $k=4$, and $m_{D_0} = k-1$ edges for $k \geq 5$.
  
  For $k = 4$ we have $n_{D_0} = \frac{k-2}{2}m_{D_0}+2 = 6$, which gives
  \[
  \alpha_\Gamma \leq \frac{2}{k-2}\cdot \frac{n_{D_0}-2}{n_{D_0}-1} = \frac{2}{4-2}\cdot \frac{6-2}{6-1} = \frac{4}{5}.
  \]
  For $k \geq 5$ we have analogous to Case 3 $n_{D_0} = \binom{k-1}{2}+2$, which gives
  \[
  \alpha_\Gamma \leq \frac{2}{k-2}\cdot \frac{n_{D_0}-2}{n_{D_0}-1} = \frac{2}{k-2} \cdot \frac{\binom{k-1}{2}}{\binom{k-1}{2}+1} = \frac{2(k-1)}{(k-1)(k-2)+2}.
  \]
  
  \begin{figure}[t]
  \centering
  \includegraphics{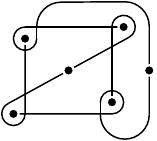}
  \hspace{4em}
  \includegraphics{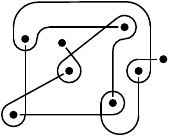}
  \caption{Smallest tight drawings for $k=4$ (left) and $k=5$ (right) in case $X = \{\Inc,\Multi\}$, i.e. incident crossings and multicrossings are forbidden. (Isolated vertices in empty cells are omitted.)}
  \label{fig:tight-MI-k45}
  \end{figure}

  \begin{figure}[t]
  \centering
  \includegraphics[page=1]{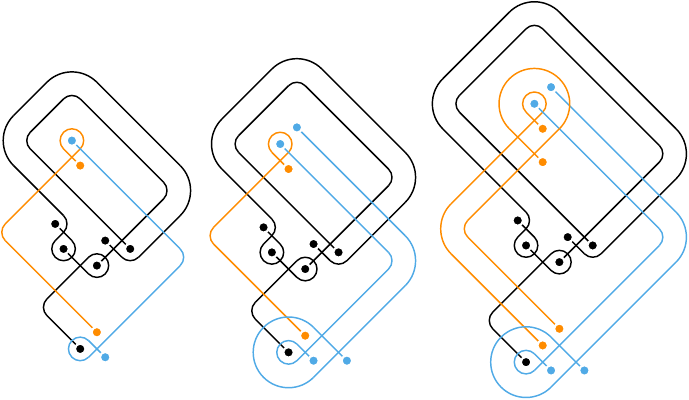}
  \caption{Smallest tight drawings for $k\geq6$ in case $X = \{\Inc,\Multi\}$, i.e. incident crossings and multicrossings are forbidden. (Isolated vertices in empty cells are omitted.)}
  \label{fig:tight-MI-k6}
  \end{figure} 
  
  On the other hand, let $D'$ be any drawing in $\Gamma_X$.
  Clearly, $D' \in \Gamma_X \subset \Gamma_{\{\Multi\}}$ for $\{\Multi\} \subset X = \{\Inc,\Multi\}$.
  However, we already argued in Case 3 that there is no drawing $D'$ in $\Gamma_{\{\Multi\}}$ with $km_{\rm x} - 1 \leq 2\mathrm{cr}$ and $1 \leq m_{\rm x} < k-1$.
  This already seals the deal for $k \geq 5$.
  
  For $k=4$, assume that $D'$ is a filled, essentially $2$-connected drawing in $\Gamma_X$.
  As before, we may assume that $2\mathrm{cr} \geq km_{\rm x}-1 = 4m_{\rm x}-1$, i.e., there are at least $2m_{\rm x}$ crossings.
  On the other hand, as multicrossings are forbidden, we have again $\mathrm{cr} \leq m_{\rm x} + \binom{m_{\rm x}}{2}$, which together implies that $m_{\rm x} \geq 3$.
  We are done if $m_{\rm x} + m_{\rm p} \geq 4 = k$.
  Otherwise, we have $m_{\rm x} = 3$ and $\mathrm{cr} \geq 2m_{\rm x} = 6$.
  Now if two of the three crossed edges were incident, they would not cross each other as $\Inc \in X$, which would give at most $3$ self-crossings and $2$ crossings of independent edges, contradicting $\mathrm{cr} \geq 6$.
  Thus we may assume that all three crossed edges are crossing themselves, pairwise crossing and pairwise independent, i.e., $n_{D'} = 6$ and $\mathrm{cr} = 6$.
  
  Now let us consider the subdrawing $H$ of the planarization of $D'$ obtained by removing all vertices of $D'$.
  I.e., $H$ is planar, connected, $|V(H)| = \mathrm{cr} = 6$, and $|E(H)| = (k-1)m_{\rm x} = 3 m_{\rm x} = 9$.
  Applying Euler's formula shows that the number of faces of $H$ is $|E(H)|-|V(H)|+2 = 9 - 6 + 2 = 5$.
  As $n_{D'} \geq 6$, some face of $H$ contains two vertices of $D'$ in its interior, showing that $m_{\rm p} \geq 1$, as $D'$ is filled.
  Thus $m_{\rm x} + m_{\rm p} \geq 3+1 = 4 = k$, as desired.
  \smallskip}{%
  \item[Case 5. $X = \{\Inc,\Multi\}$ ($\star$)]{\ }\smallskip
  }
  
  \shortandlong{%
  \item[Case 6. $X = \{\Self,\Inc,\Multi\}$]{\ \\}
	  The right of \cref{fig:k-planar-examples} (with isolated vertices added to both empty cells) and \cref{fig:tight-SIM-k7} show tight drawings $D_0$ with $m_{D_0} = k+1$ edges for $k \geq 6$.
  Analogous to Case 4 $n_{D_0} = \binom{k}{2}+1$, which gives
  \[
   \alpha_\Gamma \leq \frac{2}{k-2}\cdot \frac{n_{D_0}-2}{n_{D_0}-1} = \frac{2}{k-2} \cdot \frac{\binom{k}{2}-1}{\binom{k}{2}} = \frac{2(k+1)}{k(k-1)}.
  \]

  \begin{figure}[t]
  	\centering
  	\includegraphics[]{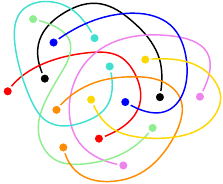}
  	\hspace{3em}
  	\includegraphics[page=1]{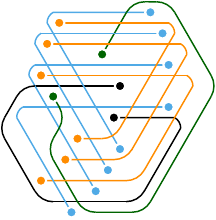}\\[\bigskipamount]
  	\includegraphics[page=2]{8-planar-9-matching}
  	\hspace{3em}
  	\includegraphics[page=3]{8-planar-9-matching}
	  \caption{Smallest tight drawings for $k\geq 7$ (for $k=6$, see \cref{fig:k-planar-examples} (right)) in case $X = \{\Self,\Inc,\Multi\}$, i.e. selfcrossings, incident crossings, and multicrossings are forbidden. 
  	Top-Left: The $8$-matching for $k=7$.
  	Top-Right: The $9$-matching for $k=8$. Bottom-Left: The $10$-matching for $k=9$.
  	Bottom-Right: The $11$-matching for $k=10$. 
  	(Isolated vertices in empty cells are omitted.)}
  	\label{fig:tight-SIM-k7}
  \end{figure}

  On the other hand, any drawing $D' \in \Gamma_X$ is also a drawing in $\Gamma_{\{\Self,\Multi\}}$ for $\{\Self,\Multi\} \subset X = \{\Self,\Inc,\Multi\}$.
  However, we already argued in Case 4 that there is no drawing $D' \in \Gamma_{\{\Self,\Multi\}}$ with $km_{\rm x}-1 \leq \mathrm{cr}$ and $m_{\rm x} < m_{D_0} = k+1$.  
\smallskip

  \item[Case 7. $X = \{\Self,\Inc,\Multi,\Homoto\}$]{\ \\}
  We can not proceed with $X = \{\Self,\Inc,\Multi,\Homoto\}$ as before, since $\Gamma_X$ is not monotone in that case.
  However, we see that the tight drawings $D_0$ in \cref{fig:k-planar-examples} (right) and \cref{fig:tight-SIM-k7} for drawing style $\Gamma_{\{\Self,\Inc,\Multi\}}$ are also in $\Gamma_X$ as there are no parallel edges and hence no homotopic edges.
  Thus 
  \begingroup
  \allowdisplaybreaks
  \begin{align*}
   \frac{m_{D_0}}{n_{D_0}-1} 
   &\geq \min\left\{ \frac{m_D}{n_D-1} \colon D \in \Gamma_X \text{ is $\Gamma_X$-saturated } \right\}\\
   &\geq \min\left\{ \frac{m_D}{n_D-1} \colon D \in \Gamma_X \text{ is filled } \right\}\\
   &\geq \min\left\{ \frac{m_D}{n_D-1} \colon D \in \Gamma_{\{\Self,\Inc,\Multi\}} \text{ is filled } \right\}\\
   &= \min\bigg\{ \frac{m_D}{n_D+c_0(D)-1} \colon D \in \Gamma_{\{\Self,\Inc,\Multi\}} \\
   &\hspace{2cm}\text{ is filled and essentially $2$-connected } \bigg\}\\
   &= \alpha_{\Gamma_{\{\Self,\Inc,\Multi\}}} = \frac{2}{k-2}\cdot \frac{n_{D_0}-2}{n_{D_0}-1} = \frac{m_{D_0}}{n_{D_0}-1}
  \end{align*}
  \endgroup
  and equality holds throughout.
  Hence, for every filled and every $\Gamma_X$-saturated drawing $D$ in $\Gamma_X$ we have $m_D \geq \alpha_{\Gamma_{\{\Self,\Inc,\Multi\}}} \cdot (n_D-1) = \frac{2(k+1)}{k(k-1)} \cdot (n-1)$.}
 \end{description}

\shortandlong{
 In Cases 1--6 we have determined exactly $\alpha_\Gamma$ for each considered drawing style $\Gamma = \Gamma_X$.
 By \cref{lem:wlog-2-connected,lem:saturated-are-filled} every $\Gamma$-saturated drawing $D$ satisfies $m_D \geq \alpha_\Gamma(n_D-1)$.
 For Case 7 we have shown this inequality directly.
 Moreover, we presented in each case a tight drawing $D_0$ attaining this bound:
 \[
  m_{D_0} = \frac{2}{k-2}(n_{D_0}-2) = \frac{2}{k-2}\cdot \frac{n_{D_0}-2}{n_{D_0}-1} \cdot (n_{D_0}-1) = \alpha_\Gamma(n_{D_0}-1)
 \]
 It remains to construct an infinite family of $\Gamma$-saturated drawings attaining this bound.
 To this end it suffices to take tight drawings with $\alpha_\Gamma(n-1)$ edges and iteratively glue these at single vertices, which always results in a tight drawing again.
 
 Formally, for vertices $v_1,v_2$ in two (not necessarily distinct) tight drawings $D_1$ and $D_2$, respectively, with $m_{D_i} = \alpha_\Gamma(n_{D_i}-1)$ for $i=1,2$, we consider the drawing $D$ obtained from $D_1,D_2$ by identifying $v_1$ and $v_2$ into a single vertex and putting $D_2$ completely inside a cell of $D_1$ incident to $v_1$.
 Then $D$ is again tight and thus $\Gamma$-saturated.
 Moreover we have $n_D = n_{D_1}+n_{D_2} -1$ and
 \[
  m_D = m_{D_1} + m_{D_2} = \alpha_\Gamma(n_{D_1}-1) + \alpha_\Gamma(n_{D_2}-1) = \alpha_\Gamma(n_D - 1).\qedhere
 \]}
\end{proof}}

\section{Bounds for Simple Graphs}
\label{sec:simplegraphs}
\latertitle{Omitted Proofs of \cref{sec:simplegraphs}}
We define a \emph{simple filled} drawing $D$ of a simple graph $G$ as a drawing in which 
any two vertices that are incident to the same cell $c$ of $D$ are connected.
In contrast to filled drawings (according to \Cref{def:filled}) the connecting edge may (partially or completely) lie outside of the boundary of $c$.
With this definition in mind,
\cref{lem:wlog-2-connected,lem:saturated-are-filled}
directly translate to the simple graph setting 
(note that $\Homoto \not\in X$ for any drawing style $\Gamma_X$ in this setting).
\cref{lem:2-connected-count} though does not translate and 
consequently neither does the bound in \cref{lem:bound-from-tight}.
We obtain the following bound on $m_D$.

\begin{figure}[b]
  \centering
  \begin{minipage}[t]{.45\textwidth}
    \centering
    \includegraphics{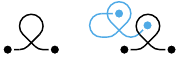}
  \end{minipage}
  \quad
  \begin{minipage}[t]{.45\textwidth}
    \centering
    \includegraphics{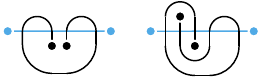}
  \end{minipage}
  \caption{Modifications of the constructions used in \cref{thm:results} for 
  $X =\emptyset$ and $X = \{\Inc\}$ on the left and $X = \{\Self\}$ and $X = \{\Self,\Inc\}$ on the right.}
  \label{fig:simple-mod}
\end{figure}

\both{\begin{lemma}[$\star$]
  For any $k$-planar simple filled and essentially $2$-connected drawing $D$
  it holds that $m_D \geq \frac{2}{k+2}(n_D - 1)$.
\end{lemma}}
\later{\begin{proof}
  Consider the planarization $P$ of $D$.
  As in the proof of \cref{lem:2-connected-count} we find that
  with $P$ being essentially $2$-connected it has has exactly
  $\#{\rm isolated} + 1$ connected components where $\#{\rm isolated}$ is the number of isolated vertices.
  Moreover, for the number of vertices and edges of $P$ it holds that $|V(P)| = n_D + {\rm cr}$ and $|E(P)| = m_D + 2{\rm cr}$,
  with ${\rm cr}$ being the number of crossings in $D$.
  Let $\#{\rm cells}$ be the number of faces of $P$.
  Since $D$ is simple filled it holds that
  $\#{\rm cells} \geq \#{\rm isolated} + 1$.
  By applying Euler's formula we obtain
  \begin{align*}
    2 &= |V(P)| - |E(P)| + \#{\rm cells} - \#{\rm isolated}\\
      &= n_D - m_D - {\rm cr} + \#{\rm cells} - \#{\rm isolated}\\
      &\geq n_D - m_D - {\rm cr} + \#{\rm isolated} + 1 - \#{\rm isolated}\\
      &= n_D - m_D - {\rm cr} + 1.
  \end{align*}
  Hence, $m_D + {\rm cr} \geq n_D - 1$ and with ${\rm cr} \leq \frac{k}{2}m_D$ we obtain the desired bound.
\end{proof}}
Consequently we get for any simple filled drawing $D \in \Gamma$ (and hence for every saturated $k$-planar drawing of a simple graph) that
\begin{align*}
  \frac{m_D}{n_D-1} \geq \min_{D'} \frac{m_{D'}}{n_{D'}-1}
 \geq \min_{D'} \frac{2}{k+2} \cdot \frac{n_{D'} - 1}{n_{D'}-1} = \frac{2}{k+2},
\end{align*}
where both minima are taken over all $k$-planar simple filled, essentially $2$-connected drawings $D' \in \Gamma$.

\begin{figure}[bt]
  \centering
  \begin{minipage}[t]{.32\textwidth}
    \centering
    \includegraphics[page=4]{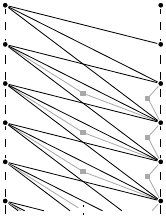}
    \label{fig:1psimple}
  \end{minipage}
  \hfill
  \begin{minipage}[t]{.32\textwidth}
    \centering
    \includegraphics[page=5]{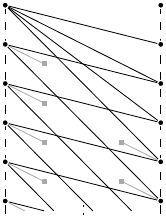}
    \label{fig:2psimple}
  \end{minipage}
  \hfill
  \begin{minipage}[t]{.32\textwidth}
    \centering
    \includegraphics[page=6]{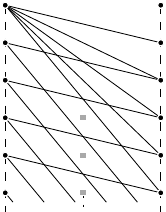}
    \label{fig:3psimple}
  \end{minipage}
  \caption{Construction for saturated simple $k$-plane drawings. 
    The dashed left and right sides of the drawings are identified.
  }
  \label{fig:simple}
\end{figure}
\later{\begin{figure}[bt]
  \centering
  \begin{minipage}[t]{.32\textwidth}
    \centering
    \includegraphics[page=1]{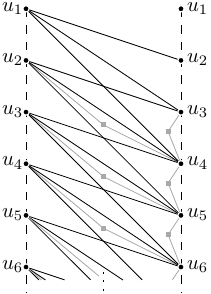}
    \label{fig:1psimple:long}
  \end{minipage}
  \hfill
  \begin{minipage}[t]{.32\textwidth}
    \centering
    \includegraphics[page=2]{simple_upperbound_cylinder}
    \label{fig:2psimple:long}
  \end{minipage}
  \hfill
  \begin{minipage}[t]{.32\textwidth}
    \centering
    \includegraphics[page=3]{simple_upperbound_cylinder}
    \label{fig:3psimple:long}
  \end{minipage}
  \caption{Construction for saturated simple $k$-plane drawings. 
    The dashed left and right sides of the drawings are identified.
  }
  \label{fig:simple:long}
\end{figure}}

Considering upper bounds on the minimum number of edges in any $\Gamma_X$-saturated $k$-planar drawing of a simple graph,
we show in the following theorem that for any drawing style $X \subseteq \{\Self,\Inc,\Multi\}$ there exist sparser drawings
than for multigraphs.
Moreover, for $X = \emptyset$ and $X = \{\Inc\}$ the resulting bound is tight.

\both{\begin{theorem}[$\star$]
  \label{thm:simpleresults}
  Let $X \subseteq \{\Self,\Inc,\Multi\}$ be a set of restrictions, and
  $\Gamma = \Gamma_X$ be the corresponding drawing style of $k$-planar drawings of simple graphs.
  For infinitely many values of $n$,
  the minimum number of edges in any $n$-vertex $\Gamma$-saturated drawing is upper bounded by
  \begin{center} \renewcommand{\arraystretch}{1.5}
    \begin{tabular}{rl}
      $\frac{2}{k + 2}(n-1)$ & \quad for $X = \{\Inc\}$ and $X = \emptyset$ and $k \geq 2$. \hfill(\cref{fig:simple-mod})\\
      $\frac{2}{k + ((k+1)\mod 2)}(n-1)$ & \quad for $X = \{\Self\}$ and $X = \{\Self,\Inc\}$ and $k \geq 4$. \hfill(\cref{fig:simple-mod})\\  
    $\frac{2}{k-1}(n-1)$ & \quad for $\Multi \in X$ and $X \subseteq \{\Self,\Inc\}$ and $k \geq 1$. \hfill(\cref{fig:simple})\\
    \end{tabular}
  \end{center}
  
\end{theorem}}
\later{\begin{proof}
  For $X = \emptyset$ and $X = \{\Inc\}$, as well as $X = \{\Self\}$ and $X = \{\Self, \Inc\}$ we modify the constructions used in \cref{thm:results}.
  In \cref{fig:simple-mod} we show the modifications.
  Taking disjoint unions of one of these drawings by placing instead of an isolated vertex
  the whole drawing again into an empty cell 
  leads a saturated $k$-planar drawing of the respective drawing style 
  with $m_D = \frac{m_{D_0}}{n_{D_0} + c_0(D_0) - 1}(n_D - 1)$ edges.
  In some sense, this is exactly the first argument of \cref{lem:bound-from-tight} in reverse, i.e.,
  instead of replacing parts of the drawing by isolated vertices, we replace an isolated vertex by a drawing.
    
  In case of $X = \emptyset$ and $X = \{\Inc\}$ this leads to $m_D = \frac{1}{2 + k - 1}(n_D-1) = \frac{2}{k+2}(n_D-1)$ many edges.
  For $X = \{\Self\}$ and $X = \{\Self, \Inc\}$ we get $m_D = \frac{2}{4 + k - 2 - 1}(n_D-1) = \frac{2}{k + 1}(n_D-1)$ if $k$ is even and
  $m_D = \frac{2}{4 + k - 3 - 1}(n_D-1) = \frac{2}{k}(n_D-1)$ if $k$ is odd. 

  For $X = \{\Self,\Inc,\Multi\}$ we modify a construction originally 
  presented by Brandenburg et al.~\cite{BrandenburgEGGHR12} for $1$-planar drawings and 
  adapted by Auer et al.~\cite{AuerBGH12} for $2$-planar drawings.
  Here, we generalize the construction to $k$-planar drawings.
  See \cref{fig:simple:long} for an illustration for the cases $k = 1$, $2$, and $3$.
  The construction is more easily imagined on a cylinder.
  We describe it here for $k \geq 3$.
  Let $D'$ be the drawing, it consists of a path on $n_{D'}$ vertices $u_1,\ldots,u_n$ with
  the vertices being laid out along a vertical line from the top to the bottom of the cylinder and
  each vertex $u_i$ with $i \not\in \{1,n\}$ being connected to $u_{i+1}$ and 
  $u_{j}$ with $j = i + k + 2$ where we do not add an edge if $u_{j}$ does not exist.
  Finally, we add all edges $u_1u_j$ for $j = 2,\ldots,k+2$ and 
  $u_ju_n$ for $j = n_{D'} - (k + 2)$.
  As a result we obtain a series of cells with no vertex on their boundaries 
  in which we add isolated vertices.
  Clearly this drawing is $k$-planar and no edge can be added without crossing another edge more than $k$ times.
  Moreover, there are no multiple-, incident-, or selfcrossings and hence $D'$ is $\Gamma_X$-saturated.
  Moreover, for $n_{D'} \geq k+3$ the drawing has $m_{D'} = 2n_{D'} + k - 3$ many edges.
  It remains to count how many isolated vertices we can add.
  For every edge $u_iu_{i+1}$ with $i = k,\ldots,n_{D'}-k-1$ bounds $k-2$ cells in which we can place an isolated vertex.
  Additionally, the edges $u_iu_{i+1}$ with $i = 3,\ldots,k-1$ and $i = n_{D'} - 4, \ldots, n_{D'} - k$
  bound $1,\ldots,k-3$ many cells in which we can place one isolated vertex each.
  In total we get that \[n_D = n_{D'} + n_{D'}(k-2) - 2k(k-2) + (k-3)(k-2).\]
  Solving for $n_{D'}$ we obtain that
  \[n_{D'} = \frac{n_D + k^2 +k -6}{k-1}.\]
  Plugging the above into $m_D = 2n_{D'} + k - 3$ we finally obtain
  \begin{align*}
    m_D = 2\frac{n_D + k^2 +k -6}{k-1} + k - 3 = \frac{2}{k - 1}(n_D - 1) + \frac{3k^2 - 2k - 9}{k-1}.
  \end{align*}\qedhere
\end{proof}}

\section{Concluding Remarks}
\label{sec:conclusions}

With respect to multicrossings, in this work we either disallowed their existence (\Multi) or did not restrict their number. 
It is possible to make a more fine-grained analysis and consider the maximum number of times that a pair of edges (or an edge with itself) is allowed to cross as a parameter $\mu$. 
Modifications of our constructions, for example 
retracing a side of each edge in the construction in \cref{fig:tight-SIM-k7} from both endpoints, 
yield tight bounds for
arbitrarily many values of $k$ and $\mu$.

Our drawings typically contain a large amount of isolated vertices.
We discussed the case that isolated vertices are not desired already in the introduction: in this case the sparsest graphs possible are matchings and saturated $k$-planar drawings of matchings indeed exist for $k\geq 6$
(see~\cref{fig:k-planar-examples,fig:tight-SIM-k7}).
These drawings are simple and hence contained in all specific drawing styles that we consider.
Disjoint unions of these drawings also yield arbitrarily large saturated drawings of matchings for any fixed $k\geq 6$.
For $k\leq 5$ saturated $k$-planar drawings of matchings do not exist provided homotopic parallel edges are allowed.
For other drawing styles and for simple graphs saturated drawings of matchings may exist also in case $k\leq 5$.
For connected graphs, our best answers are the saturated drawings of cycles depicted in \cref{fig:tight-MS-k4,fig:k-planar-examples} and (non-simple) drawings of trees in \cref{fig:tight-M-k4}.
It is an interesting question to characterize those trees that admit saturated drawings for some fixed $k$ (with respect to the drawing styles discussed here).

For simple graphs, it is a relevant open question to determine the minimum number of edges 
in a saturated $k$-planar simple drawing.
Finally, our techniques only work for fixed drawings.
It remains open to determine the min-saturated $k$-planar (abstract) graphs and the sizes of their edge sets.
 
\bibliography{k-planar-gd}

\ifbool{long}{}{%
  \clearpage
  \appendix%
  \magicappendix%
}%

\end{document}